\documentclass{amsart}
\usepackage[british]{babel}
\usepackage[OT2,T2A,T1]{fontenc}
\usepackage[utf8]{inputenc}		
\usepackage{eulervm}	
\usepackage[scr=dutchcal, bb=boondox]{mathalfa}

\usepackage{titlesec}
\titleformat{\section}
	{\centering\Large\scshape\bfseries}{\thesection.}{1em}{} 
\titleformat{\subsection}
 	{\normalfont\bfseries\scshape}{\thesubsection.}{.5em}{}
\titleformat*{\paragraph}{\bfseries}

\usepackage{titletoc}

    \titlecontents{section}
    [0em] %
    {\medskip}
    {\thecontentslabel\quad}
    {}
    {\hfill\contentspage}

    \newenvironment{acknowledgements}{%
  
  \begin{abstract}
}{%
  \end{abstract}
}

\usepackage{amsmath,amsthm,amssymb,dsfont,extarrows,amscd,amsfonts,mathtools,bm,physics}
\usepackage[v1]{subfiles}
\usepackage{hyperref}
\usepackage{mathrsfs}
\usepackage[mathscr]{euscript}
\usepackage{tensor}
\usepackage[most]{tcolorbox}
\usepackage{stmaryrd}
\usepackage[capitalize]{cleveref}
\usepackage{MnSymbol}
\usepackage{thmtools}

\usepackage{tikz}
\usepackage{tikz-cd}
\usetikzlibrary{graphs,decorations.pathmorphing,decorations.markings}
\tikzset{
    labl/.style={anchor=south, rotate=90, inner sep=.5mm}
}

\usepackage{todonotes}
\usepackage{geometry}
\usepackage[alphabetic, initials]{amsrefs}

\usepackage{titlesec}

\titleformat{\paragraph}
{\normalfont\normalsize\bfseries}{\theparagraph}{1em}{}
\titlespacing*{\paragraph}
{0pt}{2.5ex}{1ex}

\DeclareFontFamily{U}{MnSymbolC}{}
\DeclareSymbolFont{MnSyC}{U}{MnSymbolC}{m}{n}
\DeclareFontShape{U}{MnSymbolC}{m}{n}{
	<-6>  MnSymbolC5
	<6-7>  MnSymbolC6
	<7-8>  MnSymbolC7
	<8-9>  MnSymbolC8
	<9-10> MnSymbolC9
	<10-12> MnSymbolC10
	<12->   MnSymbolC12}{}
\DeclareMathSymbol{\intprod}{\mathbin}{MnSyC}{'270}

\theoremstyle{plain} 
\newtheorem{thm}{Theorem} [section]

\newtheorem{lem}{Lemma} [section]
\newtheorem{prop}{Proposition}[section]

\theoremstyle{definition}
\newtheorem{defn}{Definition}[section]
\newtheorem{ex}{Example}[section]

\newtheorem{thmdef}{Theorem/Definition}[section]

\theoremstyle{remark} 
\newtheorem{oss}{Remark} [section]

\numberwithin{equation}{section}

\newcommand{\RR}{\mathbb{R}}
\newcommand{\CC}{\mathbb{C}}
\newcommand{\ZZ}{\mathbb{Z}}

\DeclareMathOperator{\X}{\mathfrak{X}}
\DeclareMathOperator{\Lie}{\mathscr{L}}

\newcommand{\Hol}{\mathscr{O}}
\newcommand{\RS}{\mathbb{P}^{1}}
\newcommand{\Hw}{\mathscr{H}}
\newcommand{\UC}{\mathscr{C}}
\newcommand{\Fr}{\mathscr{F}}

\newcommand{\Frst}{\overset{\star}{\Fr}}

\newcommand{\EE}{\mathscr{E}}
\newcommand{\End}{\mathrm{End}}
\newcommand{\cech}{\mathfrak{A}}

\newcommand{\inp}[2]{\langle\, \bm{#1}\,,\, \bm{#2} \,\rangle}

\title{Open WDVV equations and $\bigvee$-systems}

\author[A.~Proserpio]{Alessandro Proserpio}
\address[A.~Proserpio]{
School of Mathematics and Statistics,
University of Glasgow, Glasgow, G12 8QQ, United Kingdom.}
\email{a.proserpio.1 [at] research.gla.ac.uk}

\author[I. A. B.~Strachan]{Ian A. B. Strachan}
\address[I. A. B.~Strachan]{School of Mathematics and Statistics,
University of Glasgow, Glasgow, G12 8QQ, United Kingdom.}
\email{ian.strachan [at] glasgow.ac.uk}

\begin{document}
\renewcommand{\hbar}{\hslash}

\begin{abstract}
The idea of a $\bigvee$-system was introduced by Veselov in the study of rational solutions of the WDVV equations of associativity. These are algebraic/geometric conditions on the set of covectors, and on the corresponding hyperplane arrangement. Here, this idea is generalized to open WDVV equations, which are an additional set of PDEs originating from open Gromow-Witten Theory. We develop -- for rank-one extensions -- algebraic/geometric conditions on the covectors that supplement the $\bigvee$-system to give rational solutions to the open WDVV equations. Examples, and the relation to superpotentials and to Dubrovin almost-duality, are given.
\end{abstract}

\maketitle
\tableofcontents


\section{Introduction}

The concept of a (Dubrovin)-Frobenius manifold, the geometric formulization of the WDVV equations of associativity, has played a pivotal role in unifying apparently disparate areas of mathematics and mathematical physics, from TQFTs and enumerative geometry, via algebraic geometry and singularity theory to integrable bi-Hamiltonian hierarchies. However, it has been shown that this is too rigid a structure to encompass recent developments in the field, and hence various generalization have been introduced, for example:
\begin{itemize}
\item[-] generalized Frobenius manifolds \cite{GenFrob};
\item[-] bi-flat F-manifolds \cite{ArsieBuryakLorRossi}.
\end{itemize}
The inspiration for this paper comes from another generalization: open Gromov-Witten theory \cites{horev2012opengromovwittenwelschingertheoryblowups,GKT26}.

Central to all these structures are the WDVV equations
\begin{align}\label{eq:WDVV}
\sum_{\mu,\nu=1}^n \bigl(\Fr_{\alpha \beta\mu}\eta^{\mu\nu}\,\,\Fr_{\nu \gamma \delta}-\Fr_{\alpha\gamma \mu }\eta^{\mu\nu}\,\,\Fr_{\nu \beta \delta}\bigr)=0\,.
\end{align}
Here $\Fr({\bf{t})}$ is a scalar function, known as the prepotential, and the constants $\eta^{\mu\nu}$ define a flat metric, with the coordinates $\{\bf t\}$ being the flat coordinates of this metric.

Two classes of solution, both associated with a finite irreducible Coxeter group $W$ were constructed in the 1990's. In one class, a Frobenius manifold structure was constructed on the orbit space $\mathbb{C}^N/W\,.$ The prepotential for this class is polynomial in the flat coordinates, and the scaling degrees of the variables are related to the degrees of the classical $W$-invariant polynomials. In the second class, stemming from Seiberg-Witten theory, the solution takes the compact form
\[
\Frst(\bm{z})=\tfrac{1}{4} \sum_{\alpha\in \mathfrak{R}_W} \alpha({\bm{z}})^2 \log \alpha({\bm{z}})\,,
\]
with $\mathfrak{R}_W$ being the root system of the Coxeter group $W\,.$ These two classes of solution were unified by Dubrovin's notion of \lq almost duality\rq, which provides a map between these two solutions. Schematically:
\[
\Fr\quad \overset{\rm almost}{\underset{\rm duality}{\longleftrightarrow}}\quad \Frst\,.
\]

\noindent In open Gromov-Witten theory one extends the multiplication by introducing additional vector-valued prepotentials ${\bm{\Omega}}$ (see Proposition \ref{opendef} for the precise definition) and this gives an extended multiplication of vectors and hence extensions of the WDVV-equations. So schematically:
\[
\Fr\quad \longleftrightarrow \quad \{ \Fr, {\bm{\Omega}}\}\,.
\]
Combining these ideas, one can develop almost-duality for open WDVV theory, and this was studied in \cite{openWDVVduality}.
\begin{equation}
       \begin{tikzcd}[every cell/.append style={align=center}]
   \begin{tabular}{c}
         $\Fr$ \\
          \footnotesize Weighted-homogeneous \\\footnotesize WDVV solution.{}
      \arrow[dd, leftrightarrow, "\text{duality}"]
      \end{tabular} & ${\relbar\joinrel\relbar\joinrel\relbar\joinrel\relbar\joinrel\relbar\joinrel\relbar\joinrel\relbar\joinrel\rightarrow}$ & \begin{tabular}{c}
         $\{ \Fr, {\bm{\Omega}}\}$ \\
          \footnotesize Weighted-homogeneous \\\footnotesize open WDVV solution.{}
          \arrow[dd, leftrightarrow, "\text{generalized duality}"]
      \end{tabular}\\
        & &\\
  \begin{tabular}{c}
         $\Frst$ \\
          \footnotesize Dual WDVV solution.
      \end{tabular}  & $\overset{\rm Section~\ref{mainresults}}{\relbar\joinrel\relbar\joinrel\relbar\joinrel\relbar\joinrel\relbar\joinrel\relbar\joinrel\relbar\joinrel\rightarrow}$&\begin{tabular}{c} 
             $\{ \Frst, \overset{\star}{\Omega} \}$ \\
          \footnotesize Dual open WDVV solution.
      \end{tabular} 
    \end{tikzcd}
\end{equation}

Separate to these developments, Veselov introduced the idea of a $\bigvee$-system. Here a solution of the WDVV equations is constructed by the ansatz
\begin{equation}\label{Fagain}
\Frst_{\mathfrak{A}}(\bm{z}) =\tfrac{1}{4} \sum_{\alpha\in \mathfrak{A}} \alpha({\bm{z}})^2 \log (\alpha({\bm{z}}))\,,
\end{equation}
where the root systems $\mathfrak{R}_W$ is replaced by a set of covectors $\mathfrak{A}\,.$ The WDVV equations then reduce to algebraic conditions on these vectors known as $\bigvee$-conditions. Clearly any root system of an irreducible finite Coxeter group is a $\bigvee$-system, but the idea is far more general. However, while many classes of $\bigvee$-systems have been constructed, no classification has been found.

In this paper we apply the idea of a $\bigvee$-system to the open WDVV equations, developing algebraic conditions on a set of vectors $\mathfrak{\widetilde{B}}$ so that the function
\begin{equation}
    {\overset{\star}{\Omega}}_\mathfrak{\widetilde{B}}(\bm{p}):=\sum_{\bm{\bm{\widetilde{\beta}}}\in \widetilde{{\mathfrak{B}}}} k_{\widetilde{\beta}}\,\bm{\widetilde{\beta}} (\bm{p})\,\log\bm{\widetilde{\beta}}(\bm{p})\,,\qquad\qquad \bm{p}\in \widetilde{V}\,,
\end{equation}
together with \ref{Fagain} satisfy the open WDVV-equations.

\section{Preliminaries}

\subsection{Frobenius manifolds}
\begin{defn}
    A \emph{(Dubrovin-)Frobenius manifold} with charge (or conformal dimension) $d\in\CC$ is a complex manifold $M$ equipped with:
    \begin{itemize}
    \item an $\Hol_M$-bilinear pointwise multiplication $\bullet:\X_M\otimes_{\Hol_M}\X_M\to \X_M$ on the tangent sheaf, called \emph{Frobenius product},
    \item a symmetric and non-degenerate bilinear form $\eta:\X_M\otimes_{\Hol_M}\X_M\to\Hol_M$, called \emph{Frobenius pairing} (or metric),
    \item two holomorphic vector fields $e$ and $E$, respectively called \emph{unity} and \emph{Euler vector field},
\end{itemize}
satisfying:
\begin{enumerate}
    \item [(DFM1)] If ${}^\eta\nabla$ denotes the Levi-Civita connection of $\eta\,$, for any $\hbar\in\CC\,$, the \emph{deformed connection} ${}^\hbar\nabla:={}^\eta\nabla+\hbar\,\bullet$ is flat and torsion-free.
\item [(DFM2)] For any holomorphic vector field $X\,$, the endomorphism $Y\mapsto X\bullet Y$ is self-adjoint with respect to $\eta\,$,
\item [(DFM3)] $e$ is the unity of $\bullet\,$, and it is a flat section of ${}^\eta\nabla\,$, i.e. ${}^\eta\nabla e=0\,$.
\item [(DFM4)] The Euler vector field $E$ satisfies the following:
  \[
        \begin{aligned}
            \Lie_Ee&=-e\,,&&& \Lie_E\bullet&=\bullet\,, &&& \Lie_E\eta&=(2-d)\,\eta\,.
        \end{aligned}
        \]
        Here, $\Lie$ denotes the Lie derivative.
\end{enumerate}
\end{defn}

\begin{oss}
   A Frobenius manifold is said to be semi-simple if $T_pM$ is generically a semi-simple complex algebra for $p\in M\,$. This is an open condition, and there are local coordinates $u_1,\dots, u_n$ -- called \textit{canonical} -- around any such point such that the corresponding coordinate vector fields $\partial_{u_1},\dots, \partial_{u_n}$ are idempotents \cite{Dubrovin1996}. 
\end{oss}

The condition (DFM1) encodes a series of requirements, which can be spelled out as follows: 
\begin{enumerate}
    \item [(DFM1a)] The connection ${}^\eta\nabla$ is flat. We will equivalently say that $\eta$ is a flat metric.
    \item [(DFM1b)] The multiplication $\bullet$ is associative and commutative.
    \item [(DFM1c)] For any fixed one-form $\phi\,$, the tensor field $(X,Y)\mapsto \phi(X\bullet Y)$ is symmetric, and its covariant derivative is also totally symmetric. In other words, it is a Codazzi tensor of rank two \cite{codazzi}.
\end{enumerate}
These conditions allow for a concrete local description of the Frobenius manifold structure in terms of one local function on $M\,$. Firstly, we consider the rank-three tensor field $c$ $\eta$-dual to $\bullet\,$, i.e.:
\begin{equation}
    c(X,Y,Z):=\eta(X\bullet Y\,,\,Z)\,,
\end{equation}
for any three sections $X,Y$ and $Z$ of $\X_M\,$.

Owing to (DFM1b) and (DFM2), $c$ is totally symmetric, and so is its covariant derivative, i.e. it is a Codazzi tensor of rank three. Since the curvature of $\eta$ vanishes, it follows from \cite{codazzi} that there is a local function $\Fr\in\Hol_M(U)$ in a neighbourhood $U$ of each point of $M$ such that:
\begin{equation}
    c\,\lvert_U={}^\eta\nabla^3\Fr\,.
\end{equation}
Explicitly, $c(X,Y,Z)=X(Y(Z(\Fr))-{}^\eta\nabla_XY(Z(\Fr))-(\nabla^2_{X,Y}Z)(\Fr)\,$. The function $\Fr$ is called \emph{prepotential} (or free energy) of the Frobenius manifold, and it is only defined locally up to a function in the kernel of ${}^\eta\nabla^3\,$.

The metric $\eta$ can also be expressed locally in terms of $\Fr\,$:
\begin{equation}\label{eq:etalocal}
    \eta\,\lvert_U={}^\eta\nabla^2e(\Fr)\,.
\end{equation}

Furthermore, from (DFM4), the prepotential needs to satisfy:
\begin{equation}\label{eq:quasihomogeneityF}
    E(\Fr)=(3-d)\,\Fr \quad \mod {\rm quadratic~terms\,.}
\end{equation}

Since the metric $\eta$ is flat, there are distinguished set of coordinates where the local description of the structure is particularly simple. These are those in which the Christoffel symbols of ${}^\eta\nabla$ vanish or, equivalently, where the Gram matrix that represents $\eta$ is constant. These coordinates are usually referred to as being themselves \emph{flat}.

As a function of the flat coordinates $t_1,\dots,t_n\,$, the prepotential is defined up to quadratic polynomials, and the components of $c$ are its third derivatives. The associativity of the multiplication is an algebraic constraint on the structure constants of $\bullet\,$, therefore a system of third-order PDEs for $\Fr\,$. These are precisely the WDVV equations in \cref{eq:WDVV}.

Since the identity is flat, then we can always choose a system of flat coordinates so that $e$ is one of the coordinate vector fields. Unless otherwise stated, we will make the choice $e=\partial_{t_1}\,$. In such a case, \cref{eq:etalocal} becomes $\eta_{\alpha\beta}=\Fr_{1\alpha\beta}\,$, which is a constant, non-degenerate matrix.

For the Euler vector field, we have:
\begin{lem}[\cite{Dubrovin1999}]\label{lem:EVFaffine}
    Let $M$ be a Frobenius manifold. The Euler vector field $E$ is affine with respect to ${}^\eta\nabla $, i.e. ${}^\eta\nabla ^2 E=0$.
\end{lem}

As a consequence, in a system of flat coordinates, the components of $E$ are at most linear functions. We can always choose the flat coordinates $\bm{v}$ so that the coordinate vector fields are eigenvectors of $X\mapsto{}^\eta\nabla _XE$ \cite{Dubrovin1999}. In other words, if we denote $k:=\rank{}^\eta\nabla  E\,$, $E$ is decomposed in the associated local frame as follows:
\begin{equation}\label{eq:EulerVFflatcoords}
\begin{aligned}
        E&=\sum_{\mu=1}^kd_\mu \,t_\mu\partial_{t_\mu}+\sum_{\mu=k+1}^nr_{\mu-k} \,\partial_{t_\mu}\,,
\end{aligned}
\end{equation}
for some complex numbers $d_1,\dots, d_k\in\CC^\ast\,$, $r_1,\dots, r_{n-k}\in\CC\,$. $\Lie_Ee=-e$ finally imposes $d_1=1\,$.

As a consequence, \cref{eq:quasihomogeneityF} says that $\Fr$ is a weighted-homogeneous function with weights $(d_1,\dots, d_k\,;\,r_1,\dots,r_{n-k}\,;\,3-d)$ of the flat coordinates $\bm{t}\,$, by which we mean:
\begin{equation}
    \Fr(e^{d_1c}\,t_1,\dots, e^{d_kc}\,t_k,t_{k+1}+c\,r_1,\dots,t_n+c\,r_{n-k})=e^{(3-d)c}\,\Fr(\bm{t})\mod {\rm quadratic~terms}\,,
\end{equation}
for all $c\in\CC\,$.

\noindent As anticipated, this gives a one-to-one correspondence between Frobenius manifold structures and weighted-homogeneous solutions to the WDVV equations.

\subsection{Dubrovin almost duality}
It was shown in \cite{Dub04} that, to any Frobenius manifold, one can associate a \textit{generalised} Frobenius structure, which satisfies the same axioms except for flatness of the unity and existence of an Euler vector field. For this reason, we define:
\begin{defn}
    A \textit{conformal} (or dual-type) \textit{Riemannian F-manifold} of charge $d\in\CC$ is a complex manifold $M$ equipped with:
    \begin{itemize}
        \item An $\Hol_M$ bilinear multiplication $\ast:\X_M\otimes_{\Hol_M}\X_M\to\X_M$ on the tangent sheaf,
        \item A symmetric, non-degenerate bilinear form $g:\X_M\otimes_{\Hol_M}\X_M\to\Hol_M\,$,
        \item A distinguished holomorphic vector field $\EE\in\X_M(M)\,$,
    \end{itemize}
    satisfying:
    \begin{enumerate}
        \item [(DRFM1)] For any $\hbar\in\CC\,$, the deformed connection ${}^\hbar\nabla:={}^g\nabla+\hbar\,\ast$ is flat and torsion-free.
        \item [(DRFM2)] For any holomorphic vector field $X\,$, the endomorphism $X\ast\in \Gamma(\End TM)$ is self-adjoint with respect to $g\,$.
        \item [(DRFM3)] $\EE$ is the identity of $\ast\,$, and it is a conformal Killing vector field for $g\,$, with conformal factor $1-d\,$, i.e.:
        \[
        \Lie_\EE g=(1-d)\,g\,.
        \]
    \end{enumerate}
\end{defn}
\begin{oss}
    Since the deformed connection to a dual-type Riemannian F-manifold is still flat and torsion-free, it follows that one can still uniquely associate to any such structure a solution to the WDVV equations. This solution will, however, in general fail to be weighted-homogeneous. We shall refer to such solutions as \textit{dual-type solutions}, and will typically denote them by $\Frst\,$.
\end{oss}
\begin{oss}
    Since $\EE$ is a conformal Killing vector field with conformal factor for $g\,$, it follows that the identity is an affine vector field for ${}^g\nabla\,$, with the same argument as in \cref{lem:EVFaffine}.
\end{oss}

In the case of a Frobenius manifold, as mentioned above, the associated dual-type Riemannian F-manifold structure is constructed as follows. The reader is referred to \cite{Dub04} for the details.
\begin{itemize}
\item The complex manifold upon which the structure is defined is the complement in the Frobenius manifold $M$ of the following closed subset, called \emph{discriminant}:
\begin{equation}\label{eq:discriminant}
        \Delta_M:=\bigl\{\,p\in M: \quad E_p\text{ does not admit an inverse with respect to }\bullet\,\bigr\}\,.
\end{equation}
    \item The product $\ast$ -- called \emph{dual product} -- is defined as follows:
    \begin{equation}\label{eq:dualproduct}
        X\ast Y:=E^{-1}\bullet X\bullet Y\,,
    \end{equation}
    for any two holomorphic vector fields on $M\smallsetminus\Delta_M\,$.
    \item The metric $g $ -- called \emph{intersection form} -- is defined on the cotangent bundle as follows:
    \begin{equation}\label{eq:cointform}
            g^\sharp(\alpha,\beta):=E\intprod (\alpha\bullet \beta)\,,
    \end{equation}
    for any two one-forms $\alpha$ and $\beta\,$, where $\bullet$ is the product on $\Omega_M^1$ induced by identifying $TM$ with $T^*M$ via the musical isomorphism associated to $\eta\,$. The cometric $g^\sharp$ is non-degenerate on the complement of the discriminant, and therefore defines a metric $g$ on the tangent bundle of $M\smallsetminus\Delta_M\,$.
\end{itemize}
\begin{thm}[\cite{Dub04}]
    Let $M$ be a Frobenius manifold with charge $d\,$. The dual product \cref{eq:dualproduct} and the intersection form, defined as the inverse of the cometric in \cref{eq:cointform} where non-degenerate, endow the complement of the discriminant \cref{eq:discriminant} with the structure of a dual-type Riemannian F-manifold of charge $d $ and unity $E\,$.
\end{thm}

\subsection{$\bigvee$-systems}
The notion of a $\bigvee$-system was introduced by A. Veselov in \cite{veselovcech}. These are collections of covectors on a fixed finite-dimensional real vector space $V\,$, satisfying some conditions, each producing a solution to the WDVV equations of dual type. The key observation under this construction was that the dual-type solutions associated to the Saito Frobenius manifolds on orbit spaces of Weyl groups \cite{Duborbitspaces} could all be written in a closed form that depends on the corresponding root system. Generalising this construction, then, one starts from a given set of covectors -- which determine a hyperplane arrangement in the corresponding vector space -- and associates to it a function of the same form as the ones coming from root systems. The $\bigvee$-conditions are requirements on the hyperplane arrangement making sure that such a function is a solution to the WDVV equations. Of course, the standard root systems are recovered, but $\bigvee$-systems are a strictly larger class of systems of covectors. The WDVV solutions associated to $\bigvee$-systems in the original work by Veselov are nowadays mostly known as \textit{rational solutions}, as various generalisations have been introduced, see e.g.
\cites{Fei06,FV07,FV08,Fei09,S10,SV14,FV18,GF19,AF21,SS21,FKS24}.
For the most recent results in this area, see \cite{BPV}\,.\\

To any collection of covectors $\alpha\in\mathfrak{A}\subset V^\star$ one may define define a symmetric pairing $V\times V\rightarrow \mathbb{R}$
\begin{equation}{\label{metricdef}}
\inp{X}{Y} = \frac{1}{2h} \sum_{\bm{\alpha}\in\mathfrak{A} }
h_\alpha \, \alpha({\bm{X}}) \alpha({\bm{Y}})
\end{equation}
where $h_\alpha$ and $h$ are constants (the constant $h$ is introduced separately for later convenience). We assume that:
\begin{itemize}
    \item this is non-degenerate and hence defines an inner product, or metric, on the space $V$ with inverse $\langle~,~\rangle^\star$ on $V^\star\,;$
    \item both $\pm\alpha \in \mathfrak{A}\,$.
\end{itemize}
Since $\pm \alpha\in\mathfrak{A}$ one can introduce a positive subsystem $\mathfrak{A}^{+}$. This inner metric then defines an isomorphism $\phi_\mathfrak{A}: V\rightarrow V^\star\,$, and we define
$\alpha^{\vee} := \phi_{\mathfrak{A}}^{-1}(\alpha)\,.$ We can now define a $\bigvee$-system:

\begin{defn} 
The system $\mathfrak{A}$ is a $\bigvee$-system if, for any two-plane $\Pi\subset V^\star$ containing $\alpha\,,$
\[
\sum_{\beta\in\Pi\cap\mathfrak{A}} h_\beta \,\beta(\alpha^\vee)\, \beta^\vee = \lambda \, \alpha^\vee
\]
where $\lambda\in\RR$ depends on the plane $\Pi$ and the covector $\alpha\,.$
\end{defn}
\noindent Any such a $\bigvee$-systems defines a solution of the WDVV-equations.

\begin{thm}[\cite{veselovcech}]
 Let $V$ be a finite-dimensional real vector space, and $\cech\subseteq V^*$ be a collection of covectors. The function 
 \begin{equation}\label{eq:prepotentialcech}
    \Frst_{\mathfrak{A}}(\bm{v}):=\tfrac{1}{2}\sum_{\bm{\alpha}\in\mathfrak{A}^{+}} h_\alpha \,
    {\alpha}({v}) ^2\,\log {\alpha}({v})\,,\qquad\qquad\bm{v}\in V\,.
\end{equation}
associated to $\cech$ is a (dual-type) solution to the WDVV equations if and only if $\cech$ is a $\bigvee$-system.
 \end{thm}
\noindent Such a function has the following homogeneity property
\[
\left(\sum_i v^i \frac{\partial~}{\partial t^i}\right) \Frst_{\mathfrak{A}} = 2 \Frst_{\mathfrak{A}} + 2 h\, \langle v,v\rangle\,,
\]
on using the definition (\ref{metricdef}) and where $\alpha(v) = \alpha_i v^i\,,$ with $v^i$ being the components of the vector $v\in V\,.$

 \begin{ex}[Coxeter root systems]
 As anticipated, the root system $\mathfrak{R}_W$ of a finite irreducible Coxeter group $W$ is a $\bigvee$-system since, if
 $\mathfrak{A}=\mathfrak{R}_W\,,$ the solution \cref{eq:prepotentialcech} is known to be the dual solutions to the Saito-Frobenius manifold on orbit space $\mathbb{C}^N/W\,.$
\end{ex}
Such $\bigvee$-systems have many interesting properties, for example, a subsystem of a $\bigvee$-system is also a $\bigvee$-system, and the restriction of a $\bigvee$-system $\mathfrak{A}$ to a subspace is also a $\bigvee$-system \cite{FV08}.

\subsection{Extensions of Frobenius manifolds}\label{opendef}
For a number of more modern applications, the notion of a fully-fledged Frobenius manifold has often proven to be too rich. Various generalisations have arisen over the years. In particular, we are here interested in the following:
\begin{defn}
    An \emph{F-manifold with (compatible) flat connection} is a complex manifold $M$ equipped with 
    \begin{itemize}
\item An $\Hol_M$-bilinear multiplication $\bullet:\X_M\otimes_{\Hol_M}\X_M\to\X_M$ on the tangent sheaf,
\item An affine connection $\nabla:\X_M\to\Omega_M^1\otimes_{\Hol_M}\X_M\,$,
        \item A distinguished holomorphic vector field $e\in\X_M(M)\,$,
        \end{itemize}
    satisfying:
    \begin{enumerate}
           \item [(FMC1)] For any $\hbar\in\CC\,$, the \emph{deformed connection} ${}^\hbar\nabla:=\nabla+\hbar\,\bullet$ is flat and torsion-free.
        \item [(FMC2)] At any $p\in M\,$, $e_p$ is the identity of $\bullet_p\,$.
    \end{enumerate}
        If, further, $\nabla e=0\,$, we will say that $M$ is \emph{normalised}.
\end{defn}
\begin{oss}
    In other words, in an F-manifold with flat connection, the connection needs not be metric, and there is no canonical notion of an Euler vector field. In particular, Frobenius manifolds are normalised F-manifolds with connection, $\nabla$ being the Levi-Civita connection of $\eta\,$.
\end{oss}
\begin{oss}
    Since, it this work, we shall always assume that $\nabla$ is flat, we shall drop the adjective \emph{flat} in F-manifold with flat connection.
\end{oss}
\begin{oss}
Normalised F-manifolds with connection are usually called \emph{flat F-manifolds}.
\end{oss}
\begin{oss}
    A dual-type Riemannian F-manifold is an example of a non-normalised F-manifold with connection.
\end{oss}
Since the tangent sheaf to an F-manifold with connection $M$ is still generated, as an $\Hol_M$-module, by flat sections of $\nabla\,$, then distinguished systems of flat coordinates can still be defined around every point. Furthermore, (FMC1) still implies a weaker notion of potentiality for the multiplication, namely there exists a local holomorphic \textit{vector field} $\Fr$ in a neighbourhood of each point -- called \textit{vector (pre)potential} -- such that:
\begin{equation}
    X\bullet Y=\nabla^2_{X,Y}\Fr\,,
\end{equation}
for any pair of local holomorphic vector fields $X$ and $Y\,$. Clearly, a vector potential is only determined up to a vector field in $\ker\nabla^2\,$. In a flat frame, this is a vector field whose components are linear functions of the flat coordinates.

Since the connection is flat, this form automatically ensures commutativity of the multiplication, whereas associativity is a non-trivial requirement. In a system of flat coordinates for $\nabla\,$, this is a system of second order PDEs for the coordinate functions of $\Fr\,$, which generalise the ordinary WDVV equations, called \emph{oriented WDVV equations} \cite{maninflatfmanifolds}. Due to their algebraic nature in this geometric setting, WDVV equations and their generalisations are often referred to as \textit{associativity equations}. \\

Building on our previous work \cite{openWDVVduality}, in this paper we are interested in the following construction of F-manifolds with connection:
\begin{defn}
    Let $M$ be an F-manifold with connection. An \emph{extension of $M$} is an F-manifold with connection $\widetilde{M}$ with a fibration $\pi:\widetilde{M}\to M$ and an Ehresmann connection $\vartheta:\pi^*(TM)\to T\widetilde{M}$ such that:
    \begin{enumerate}
        \item [(EFC1)] for any $p\in \widetilde{M}\,$, $(\pi_*)_p:T_p\widetilde{M}\to T_{\pi(p)}M$ is a homomorphism of complex algebras;
        \item [(EFC2)] $\widetilde{\nabla}\circ\vartheta=\pi^*\nabla\,$.
    \end{enumerate}

    The codimension of $M$ in $\widetilde{M}$ is called \emph{rank} of the extension, and we shall denote it by $r\in\ZZ_{>0}\,$. 
\end{defn}
\begin{oss}
    We let $M$ be a flat F-manifold as, in the following, we shall crucially consider the extension of the dual structure to a Frobenius manifold, which is in general not normalised.
\end{oss}
For further details, the reader is referred to \cites{alcoladophd,openWDVVduality}. In this work, we shall restrict to the rank-one case. Therefore, it is understood that, when we say \textit{extension of an F-manifold with connection}, we mean rank-one extension. More specifically, we will be interested in the case where the underlying F-manifold with connection is the Dubrovin dual of a Frobenius manifold.

The reason why this construction is interesting, from the potential-theory point of view, is that, if one assumes that the structure on the given F-manifold with connection is known, and looks for extensions, then some of the oriented WDVV equations will automatically be satisfied, whereas the remaining ones constitute a system of PDEs which is significant on its own. In particular, its solutions give generating functions for the open genus-zero Gromow-Witten invariants of a real symplectic manifold, as discussed in \cite{horev2012opengromovwittenwelschingertheoryblowups}. For this reason, this kind of associativity equations is usually called \emph{open WDVV equations}.

In order to see this, we firstly notice the following obvious consequences of (EFC2):
\begin{prop}
    Let $\pi:\widetilde{M}\to M$ be an extension of the F-manifold with connection $M$. If $z_1,\dots, z_n$ is a system of flat coordinates for $M$ on the open set $U$, then there exists an \emph{adapted system of flat coordinates} $z_1,\dots, z_n,x$ for $\widetilde{M}$ on (an open subset of) $\pi^{-1}(U)$.
\end{prop}
\begin{prop}
Let $\pi:\widetilde{M}\to M$ be an extension of the F-manifold with connection $M$, let $\Fr$, $\widetilde{\Fr}$ denote the vector prepotentials of $M$ and $\widetilde{M}$ respectively on the open subsets $U$ and $\pi^{-1}(U)$. There exists a vector field $\Omega$ on (a subset of) $\pi^{-1}(U)$ -- called \emph{extended (pre)potential} -- such that:
     \begin{itemize}
         \item $\Omega\in\ker\pi_*$.
         \item $\widetilde{\Fr}=\Fr\circ\pi+\Omega$, up to affine vector fields.
     \end{itemize}
\end{prop}

Let us, now, fix a system of adapted flat coordinates $z_1,\dots, z_n,x$ on $\widetilde{M}\,$. We are going to employ Greek indices ranging from $1$ to $n\,$. Given the vector potential $\widetilde{\Fr}=\Fr\circ\pi+\Omega$ on the extension in the adapted coordinate system, the multiplication table is then given by:
\begin{equation}\label{eq:multtableextension}
    \begin{aligned}
    \pdv{z_\mu}\,\widetilde{\bullet}\,\pdv{z_\nu} &=\sum_{\alpha=1}^n\pdv[2]{\Fr_\alpha}{z_\mu}{z_\nu}\,\pdv{z_\alpha}+\pdv[2]{\Omega}{z_\mu}{z_\nu}\,\pdv{x}\,,\\
    \pdv{z_\mu}\,\widetilde{\bullet}\,\pdv{x}&=\pdv[2]{\Omega}{z_\mu}{x}\,\pdv{x}\,,\\
    \pdv{x }\,\widetilde{\bullet}\,\pdv{x}&=\pdv[2]{\Omega}{x}\,\pdv{x}\,.
\end{aligned}
\end{equation}
For the sake of simplicity, we shall from now on denote $\partial_x\equiv '$ and $c\indices{^\alpha_{\mu\nu}}:=\partial_\mu\partial_\nu \Fr_\alpha\,$. The associativity equations are, then \cite{alcoladophd}:
\begin{equation}\label{eq:openWDVV}
    \begin{aligned}
        \sum_{\alpha=1}^nc\indices{^\alpha_{\mu\nu}}\pdv[2]{\Omega}{z_\alpha}{z_\rho}+\pdv{\Omega'}{z_\rho}\pdv[2]{\Omega}{z_\mu}{z_\nu}&=\sum_{\alpha=1}^nc\indices{^\alpha_{\nu\rho}}\pdv[2]{\Omega}{z_\alpha}{z_\mu}+\pdv{\Omega'}{z_\mu}\pdv[2]{\Omega}{z_\nu}{z_\rho}\,,\\
        \sum_{\alpha=1}^nc\indices{^\alpha_{\mu\nu}}\pdv{\Omega'}{z_\alpha}+\Omega''\pdv[2]{\Omega}{z_\mu}{z_\nu}&=\pdv{\Omega'}{z_\mu}\pdv{\Omega'}{z_\nu}\,.
    \end{aligned}
\end{equation}
These are the (rank-one) \emph{open WDVV equations}.
\begin{oss}
    In the rank-one case, fixed an adapted flat coordinate system, the extended prepotential is simply one local function on $\widetilde{M}\,$.
\end{oss}
\begin{oss}
   In the original work \cites{  horev2012opengromovwittenwelschingertheoryblowups}, where open WDVV equations appeared, and in most of the subsequent papers \cites{BCT1,BBvirasoro,BuryakOpenAD,PSTInttheory,dealmeida2025openhurwitzflatf}, this name is used to refer to the system \ref{eq:openWDVV} when $\Fr$ solves the ordinary -- or, in this context, closed -- WDVV equations. This will be the only case we are also going to consider in this work, even though there is geometrically no reason why the open WDVV equation should be only defined in such case. 
     \end{oss}
    The following remark is crucial in the study of solutions of the open WDVV equations, and it is generalised to arbitrary rank \cite{openWDVVduality}:
\begin{lem}[\cite{alcoladophd}]\label{lem:auxiliaryextension}
    If $\Omega''\neq 0$ and $\Omega$ is a solution to the second set of equations in \cref{eq:openWDVV}, then it also solves the former. In such a case, the extension is called \emph{auxiliary}.
\end{lem}
\begin{proof}
    Since $\Omega''\neq 0$, we can multiply the first equation in \cref{eq:openWDVV} by $\Omega''$ and use the second one. The terms that do not cancel out give the associativity equations for $\Fr$.
\end{proof}
Auxiliary extensions of F-manifolds with connection were fully classified in \cite{alcoladophd}.

\subsection{Landau-Ginzburg models and associated open WDVV solutions}
Building on the theory of primitive forms for miniversal deformations of hypersurface singularities \cites{Saito81,Sai83}, in \cite[Appendix I]{Dubrovin1996} and \cite{Dubrovin1999}, it is shown that the geometry of a semi-simple Frobenius manifold can be captured in the singularity theory of a family of meromorphic functions, called \emph{Landau-Ginzburg superpotentials}. More precisely, following \cite{Dubrovin1999}, we pose the following:
\begin{defn}
    A \emph{genus-$g$ Landau-Ginzburg model} (or B-model) for a semi-simple Frobenius manifold $M$ is a pair $(\lambda,\omega)$ such that:
    \begin{itemize}
        \item $\lambda:M_g\to \RS$ is a meromorphic function on a line bundle over $M$ whose fibres are isomorphic to an open domain $D_g$ of a closed Riemann surface of genus $g\,$ (possibly depending on the point in $M$), called \emph{Landau-Ginzburg superpotential}.
        \item $\omega$ is a (possibly multi-valued) meromorphic differential on $D_g\,$, called \textit{primary differential} (or primary form).
    \end{itemize}
The pair needs to satisfy:
\begin{enumerate}
        \item [(LG1)] $\lambda$ restricted to each fibre of $M_g$ only has non-degenerate critical points, and its critical values in $D_g$ are canonical coordinates for $M\,$.
        \item [(LG2)] The following expressions for the Frobenius metric $\eta$ on $M$ and the tensor field $c\,$, which is $\eta$-dual to the multiplication, hold:
        \begin{equation} \label{eq:residueformulaeLG}
            \begin{split}
                \eta(X,Y)&=\sum_{x\in\mathrm{Cr}(\lambda)}\Res_x\bigl\{X(\lambda)\,Y(\lambda)\,\tfrac{\omega^2}{\dd\lambda}\bigr\}\,,\\
     c(X,Y,Z)&=\sum_{x\in\mathrm{Cr}(\lambda)}\Res_x\bigl\{X(\lambda)\,Y(\lambda)\,Z(\lambda)\,\tfrac{\omega^2}{\dd\lambda}\bigr\}\,.
            \end{split}
        \end{equation}              
        Here, $\mathrm{Cr}(\lambda)$ is the set of critical point of $\lambda\,$.
        \item [(LG3)] If $n=\dim M\,$, there exist cycles $\gamma_1,\dots, \gamma_n$ in $D_g$ such that a system of flat coordinates for the extended Dubrovin connection (see \cite{Dubrovin1999}) is given by the twisted periods:
            \[
\widetilde{v}_\mu:=\int_{\gamma_\mu}e^{\hbar\,\lambda}\,\omega\,,\qquad \mu=1,\dots, n\,.
        \]
    \end{enumerate}
\end{defn}
\begin{oss}
    In this framework, the Landau-Ginzburg superpotential takes the role of the miniversal deformation, and the primary differential replaces the primary form from Saito theory.
\end{oss}
In other words, $M$ is identified with a subset of a moduli space of meromorphic functions on surfaces of fixed genus. 
\begin{defn}
    Let $g\in\ZZ_{\geq 0}$ and $\bm{n}\in\ZZ_{\geq 0}^k$. The \emph{Hurwitz space} $\Hw_{g\,;\,\bm{n}}$ is the space of equivalence classes of pairs $(C_g,\lambda)\,$, where $C_g$ is a genus-$g$ closed Riemann surface and $\lambda:C_g\to\RS$ is a meromorphic function of $C_g$ with $k$ poles of order $n_1+1,\dots, n_k+1$ respectively. The equivalence relation is defined by the action of $\mathrm{Aut}(C_g)$ on the source curve. In other words, two pairs $(C,\lambda)\,$, $(D,\mu)$ are identified if $C\cong D$ and there exists an automorphism $f$ of $C$ such that $\mu=\lambda\circ f\,$.
\end{defn}
Owing to Riemann's existence Theorem, Hurwitz spaces are irreducible, quasi-projective complex varieties, whose dimension can be easily computed using Riemann-Hurwitz formula:
\begin{equation}
    d_{g,\bm{n}}:=\dim\Hw_{g\,;\,\bm{n}}=2(g-1)+2k+n_1+\dots+n_k\,.
\end{equation}
Hurwitz spaces are naturally equipped with a tautological line bundle, whose fibre is a Riemann surface of the given genus:
\begin{defn}[\cite{RomanoDoubleHurwitz}]
    The \emph{universal curve} $\UC_{g\,;\,\bm{n}}$ over the Hurwitz space $\Hw_{g\,;\,\bm{n}}$ is the fibre bundle over the latter whose fibre over a point $[(C,\lambda)]$ is the curve $C$ with marked points given by the poles of $\lambda\,$.
\end{defn}
The Landau-Ginzburg superpotential is a gluing of branched coverings of $\RS$ parametrised by points in a Hurwitz space -- or a subset of. Hence, it is a meromorphic function defined on (a subset of) the universal curve \cite{RomanoDoubleHurwitz}.


In order for the metric and multiplication in \cref{eq:residueformulaeLG} to define a Frobenius manifold, the primary differential has to satisfy some admissibility conditions in terms of $\lambda\,$. These have been studied in \cite[Lecture 5]{Dubrovin1996}. Primary differentials come in five families, and there as many distinct ones on a given Hurwitz space as its dimension. 

As noted e.g. in \cite{BvG22}, the $\lambda$-admissibility endows the universal curve with additional structure. In particular, it defines an Ehresmann connection on $\UC_{g\,;\,\bm{n}}\to\Hw_{g\,;\,\bm{n}}$ by taking the horizontal distribution to the be the one spanned by $\omega$ itself. This gives a canonical way to lift vector field on the Hurwitz space to the universal curve. Hence, for any vector field $X$, we have a differential operator $\delta_X^\omega:H^0(\UC_{g\,;\,\bm{n}}\,,\,\Hol_{\UC_{g\,;\,\bm{n}}})\to H^0(\UC_{g\,;\,\bm{n}}\,,\,\mathcal{K}_{\UC_{g\,;\,\bm{n}}})\,$, $\mathcal{K}_{\UC_{g\,;\,\bm{n}}}$ being the sheaf of meromorphic functions on $\UC_{g\,;\,\bm n}$ \cites{BvG22}. Since we will always be working with a fixed primary differential, from now on, we shall denote $\delta^\omega_X$ simply by $X\,$. This explains the meaning of $X(\lambda)$ in \cref{eq:residueformulaeLG}.

As for the identity and Euler vector field, they are respectively defined to be the generators of the translations and rescalings in the stabiliser of $\infty$ for the action of the group of M\"{o}bius transformations on the target $\RS\,$. In other words, they are the generators of the actions on the superpotential $\lambda\mapsto\lambda+a$ and $\lambda\mapsto b\,\lambda$ for $a\in\CC$ and $b\in\CC^\ast\,$.\\

Finally, for the Dubrovin dual structure, 
we recall that the discriminant coincides with the subset where at least one of the critical values of $\lambda$ vanish \cites{Dubrovin1996,Dub04}.

The intersection form is, then, shown to be given by a similar residue expression:
\begin{equation}
    g(X,Y)=\sum_{x\in\mathrm{Cr}(\lambda)}\Res_x\bigl\{X(\log\lambda)\,Y(\log\lambda)\,\tfrac{\omega^2}{\dd\log\lambda}\bigr\}\equiv\sum_{x\in\mathrm{Cr}(\lambda)}\Res_x\bigl\{\tfrac{1}{\lambda}X(\lambda)Y(\lambda)\tfrac{\omega^2}{\dd\lambda}\bigr\}\,.
\end{equation}

The relation between Landau-Ginzburg models and extensions of F-manifolds with connection has been highlighted in the works \cites{dealmeida2025openhurwitzflatf,openWDVVduality}. In particular, it is shown that there are two open WDVV solutions associated to any Landau-Ginzburg model, which give an extension to the universal curve of the Frobenius manifold and of its Dubrovin dual structure respectively. In particular, the first solution is going to be weighted-homogeneous with respect to an affine vector field on the universal curve, which plays the role of an eventual identity between the two F-manifold structures in the sense of \cites{maninflatfmanifolds,davidstrDubduality}. The second solution will be called of \textit{dual-type}, and it will in general fail to be weighted-homogeneous. This is summarised in the following diagram:
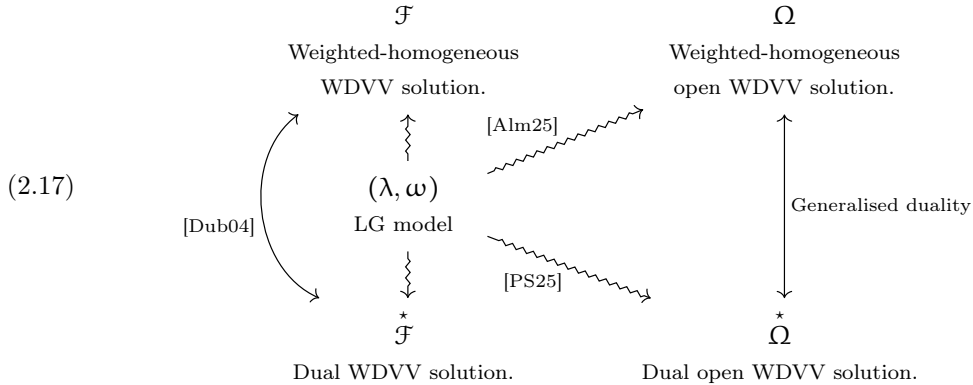
\begin{equation}
       \begin{tikzcd}[every cell/.append style={align=center}]
   \begin{tabular}{c}
         $\Fr$ \\
          \footnotesize Weighted-homogeneous \\\footnotesize WDVV solution.
      \end{tabular} &  \begin{tabular}{c}
         $\Omega$ \\
          \footnotesize Weighted-homogeneous \\\footnotesize open WDVV solution.
      \end{tabular}\arrow[dd, leftrightarrow, "\text{Generalised duality}"]\\
       \begin{tabular}{c}
          $(\lambda,\omega)$\\
         \footnotesize LG model 
        \end{tabular} \arrow[ur,rightsquigarrow, "\text{\cite{dealmeida2025openhurwitzflatf}}"]\arrow[u, rightsquigarrow]\arrow[d, rightsquigarrow]&\\
  \begin{tabular}{c}
         $\Frst$ \\
          \footnotesize Dual WDVV solution.
      \end{tabular} \arrow[uu, bend left=60, leftrightarrow, "\text{\cite{Dub04}}"]  &\begin{tabular}{c}
         $\overset{\star}{\Omega}$ \\
          \footnotesize Dual open WDVV solution.
      \end{tabular} \arrow[lu, leftsquigarrow, "\text{\cite{openWDVVduality}}"] 
    \end{tikzcd}
\end{equation}
More concretely, the two open WDVV solutions $\Omega$ and $\overset{\star}{\Omega}$ are constructed explicitly from the Landau-Ginzburg model as follows:
\begin{thm}[\cites{dealmeida2025openhurwitzflatf,openWDVVduality}]\label{HuwitzTheorem}
    Let $\Hw_{g;\bm{n}}^\omega$ be a Hurwitz Frobenius manifold of dimension $n$ with Landau-Ginzburg superpotential $\lambda$ and primary form $\omega$, $(\Hw_{g;\bm{n}}^{\omega})^\star$ denote its Dubrovin dual. Suppose that there exists a system of flat coordinates for both deformed connections ${}^\eta\nabla+\hbar\,\bullet$ and ${}^g\nabla+\hbar\,\ast$ associated to the Frobenius manifold and its dual respectively, given respectively by twisted periods:
    \[
    \begin{aligned}
        \widetilde{t}_\gamma&:=\int_\gamma e^{\hbar\,\lambda}\,\omega\,,&&& \widetilde{z}_\gamma&:=\int_\gamma \lambda^\hbar\,\omega\,,
    \end{aligned}
    \]
    where $\gamma\in \Lambda^\ast(z)$ as defined in \cite{MR2429320}. 

    There exist two unique F-manifolds with connection $\UC^\omega_{g;\bm{n}}$ and $(\UC^\omega_{g;\bm{n}})^\star$ on (an open subset of) the universal curve $\UC_{g;\bm{n}}$ that extend $\Hw_{g;\bm{n}}^\omega$ and $(\Hw_{g;\bm{n}}^{\omega})^\star$ respectively in such a way that:
\begin{itemize}
    \item if $t_1,\dots, t_n$ and $z_1,\dots, z_n$ are a system of flat coordinates for $\Hw_{g;\bm{n}}^\omega$ and $(\Hw_{g;\bm{n}}^{\omega})^\star$ respectively on an open set $\,U$, then let $x$ be a coordinate on the fibres over the points of $U$ such that $\omega=\dd{x}\,$. Then $t_1,\dots, t_n,x$ and $z_1,\dots,z_n,x$ are adapted systems of flat coordinates on an open subset of $\pi^{-1}(U)$ on the universal curve for the two F-manifold with connections respectively;
    \item in the given adapted system, the extended prepotentials $\Omega$ and $\overset{\star}{\Omega}$ respectively satisfy:
    \[
    \begin{aligned}
            \Omega_x&=\lambda\,,&&& \overset{\star}{\Omega}_x&=\log\lambda\,.
    \end{aligned}
    \]
\end{itemize}
\end{thm}
\begin{oss}
    The equalities between $\Omega_x\,$, $\overset{\star}{\Omega}_x$ and $\lambda\,$, $\log\lambda$ respectively are to be understood up to constants. This is a consequence of the fact that $\Omega$ and $\overset{\star}{\Omega}_x$ are only defined up to a linear function of the corresponding flat coordinates.
    \end{oss}
\begin{oss}
    Notice that, if, conversely, one knows a (dual-type) open WDVV solution, then the previous Theorem gives a candidate for a Landau-Ginzburg superpotential for the Frobenius manifold that the open WDVV solution extends. This is the approach we are going to discuss in \cref{sec:superpotentials}.
\end{oss}
Now, in this case, we are particularly interested in the case where the underlying Frobenius manifold is a Saito Frobenius manifold on the orbit space of a Coxeter group. As we have discussed, the dual solution is of the form in \cref{eq:prepotentialcech}, $\cech$ being the corresponding root system. Since these structures all admit a Landau-Ginzburg model, one can therefore associate them with a dual-type solution to the open WDVV equations. Such solution is computed explicitly in type-$A$ in \cite[Section 6.1]{openWDVVduality} and the functional form will be given in the next section (see equation \cref{eq:omegaprepotentialcech}).

In this work, we generalize this and derive geometric conditions that, if satisfied, ensure that to any $\bigvee$-system $\cech\,$, one can associate an \textit{open $\bigvee$-system} $\widetilde{\mathfrak{B}}$ in such a way that a function $\Omega^{\widetilde{\mathfrak{B}}}$  is a dual-type solution to the open WDVV equations, and we geometrically characterise the covectors in the open $\bigvee$-system in terms of the Coxeter geometry of weights and roots.

\section{Open $\bigvee$-systems}\label{mainresults}
Let $V$ be an $n$-dimensional vector space, and $\cech\subseteq V^*$ be a $\bigvee$-system \cite{veselovcech}. We, then, know that there exists an associated dual-type solution $\Frst_\cech $ to the WDVV equations, as defined in \cref{eq:prepotentialcech}. We let $\widetilde{V}$ be a finite-dimensional vector space with a short exact sequence of vector spaces
 \[
        \begin{tikzcd}
                        0\arrow[r ]&\RR\arrow[r] & \widetilde{V}\arrow[r] & V\arrow[r]& 0\,.
        \end{tikzcd}
    \]
We fix a splitting of the sequence, i.e. an isomorphism $\widetilde{V}\cong \RR\oplus V\,$, and write ${\bm{p}}=(\bm{x},{\bm{z}})\in\widetilde{V}$ where $\bm{x}\in\mathbb{R}$ and ${\bm{z}}\in V\,.$

The problem is construction of a system of covectors $\widetilde{\mathfrak{B}} $ in the dual space of $\widetilde{V}$ such that the function:
\begin{equation}\label{eq:omegaprepotentialcech}
    \overset{\star}{\Omega}_{\widetilde{\mathfrak{B}}}(\bm{p}):=\sum_{\bm{\bm{\widetilde{\beta}}}\in \widetilde{{\mathfrak{B}}}} k_{\widetilde{\beta}}\,\bm{\widetilde{\beta}} (\bm{p})\,\log\bm{\widetilde{\beta}}(\bm{p})\,,\qquad\qquad \bm{p}\in \widetilde{V}\,,
\end{equation}
is a solution to the open WDVV equations associated to the solution \cref{eq:prepotentialcech} of the closed WDVV equations. We shall call $\widetilde{\mathfrak{B}}$ an \emph{open $\bigvee$-system} associated to the $\bigvee$-system $\mathfrak{A}$. 

The vectors in $\widetilde{V}$ are characterised by whether they have a zero or non-zero projection onto $\RR\,.$ Here we just consider covectors with a non-zero projection and write, through a redefinition of the constants $k_{\tilde{\beta}}\,$,
\[
\widetilde{\mathfrak{B}}\ni\tilde\beta = (1, - \beta)\,, \qquad \beta\in\mathfrak{B} \subset V^\ast\,,
\]
(the minus sign is for future convenience) with corresponding constants $k_\beta\,.$ For notational convenience the we now just use $\Frst$ and $\overset{\star}{\Omega}$ instead of $\Frst_\cech$ and 
$\overset{\star}{\Omega}_{\widetilde{\mathfrak{B}}}$ and, with a slight abuse of notation, the set $\mathfrak{B}$ will just be referred to as the open $\bigvee$-system. With this, \cref{eq:omegaprepotentialcech} becomes:


\begin{equation}\label{eq:omegaprepotentialcech2}
    \overset{\star}
    {\Omega}_{\widetilde{\mathfrak{B}}}(x,\bm{z}):=\sum_{\bm{\bm{{\beta}}}\in \mathfrak{B}} 
    k_{\beta}\,
        (x-\bm{{\beta}} (\bm{z}))\,
    \log(x-\bm{{\beta}}(\bm{z}))\,.
\end{equation}
It is also natural to assume a \lq centre of mass\rq~condition $\sum_{\beta\in\mathfrak{B}} k_\beta \beta=0\,.$

Notice that, according to \cref{lem:auxiliaryextension}, it is enough to show that $\Omega$ satisfies
\begin{equation}\label{shortform}
\overset{\star}{\Omega}_{xx}
\pdv[2]{{\overset{\star}{\Omega}}}{z_\mu}{z_\nu}-\pdv{\overset{\star}{\Omega}_x}{z_\mu}
\pdv{\overset{\star}{\Omega}_x}{z_\nu}+ 
\sum_{\alpha=1}^nc\indices{^\alpha_{\mu\nu}}
\pdv{\overset{\star}{\Omega}_x}{z_\alpha}=0\,,
\end{equation}
where $c\indices{^\alpha_{\mu\nu}}$ are the structure constants defined by $\Fr$.

Substituting \cref{eq:omegaprepotentialcech2} and \cref{eq:prepotentialcech} into \cref{shortform} yields (after a polarization argument)
\[
\frac{1}{2}\sum_{\beta,\beta_\circ \in \mathfrak{B}} 
k_\beta k_{\beta_\circ} \frac{(\beta_\circ-\beta)_i (\beta_\circ-\beta)_j}{(x-\beta_\circ(\bm{z}))(x-\beta(\bm{z}))} - 
\sum_{\beta_\circ\in\mathfrak{B}} \frac{k_{\beta_\circ}}{(x-\beta_\circ(\bm{z}))}
\sum_{\alpha\in\mathfrak{A}^{+}} \frac{h_\alpha \alpha_i \alpha_j \langle\alpha,\beta_\circ \rangle^{*}}{\alpha({\bm{z})}}=0.
\]
Since there is no double pole this may be written as
\[
\sum_{{\beta_\circ}\in\mathfrak{B}} \frac{k_{\beta_\circ}}{x-\beta_\circ(z)}
\Bigg\{
\sum_{\beta\neq\beta_\circ}
k_\beta 
\frac{(\beta_\circ-\beta)_i (\beta_\circ-\beta)_j}{(\beta-\beta_\circ)(\bm{z})} - 
\sum_{\alpha\in\mathfrak{A}^{+}} \frac{h_\alpha \alpha_i \alpha_j \langle\alpha,\beta_\circ \rangle^{*}}{\alpha({\bm{z})}}
\Bigg\}=0
\]
or, 
\[
\sum_{\beta\neq\beta_\circ}
k_\beta 
\frac{(\beta_\circ-\beta)_i (\beta_\circ-\beta)_j}{(\beta-\beta_\circ)(\bm{z})} - 
\sum_{\alpha\in\mathfrak{A}^{+}} \frac{h_\alpha \alpha_i \alpha_j \langle\alpha,\beta_\circ \rangle^{*}}{\alpha({\bm{z})}}
=0\,,\qquad {\rm for~all~} \beta_\circ\in\mathfrak{B}\,.
\]
Contracting with $z^j$ yields the condition
\begin{eqnarray*}
    \sum_{\beta} k_\beta \, (\beta_\circ - \beta) & = & \sum_{\alpha\in\mathfrak{A}} h_\alpha \langle\alpha,\beta_\circ \rangle^{*} \, \alpha\,,\\
    & = & h \, \beta_\circ
\end{eqnarray*}
and with the centre-of-mass condition we arrive at the normalization condition $\sum_\beta k_\beta = h\,.$

Contracting with $dz^j$ yields a $V^\star$-valued 1-form condition
\begin{equation}\label{keyequation}
\sum_{\beta\in\mathfrak{B}\setminus \{ \beta_\circ\}}
k_\beta \,(\beta_\circ-\beta) \,d\log[ (\beta_\circ-\beta)(\bm{z})] =
\sum_{\alpha\in\mathfrak{A}^{+}} h_\alpha \, \,\alpha \langle\alpha,\beta_\circ \rangle^{*} \, d
\log[ \alpha({\bm{z}})]
\end{equation}
for all $\beta_\circ\in\mathfrak{B}\,.$ Note, the sum on the left-hand-side can equally well be described in terms of the differences $\beta_\circ-\beta\,.$

To analyse this equation we use the following:

\begin{itemize}

    \item[(i)] Suppose $\alpha$ is a non-zero covector. This defines a hyperplane
    \[
    H_\alpha = \{ \bm{z} \in V \,|\, \alpha({\bm{z}})=0 \}\,.
    \]
    Since proportional covectors define the same hyperplane it is useful to define an equivalence relation
    \[
    \alpha \sim \alpha^\prime {\iff} H_\alpha=H_{\alpha^\prime}
    \]
    and hence all covectors in the equivalence class $[\alpha]$ define the same hyperplane.
    \item[(ii)] Given an $n$-form $\omega$ with a simple pole along a hyperplane $H_\alpha\,,$ the residue (an $(n-1)$-form) is defined by the expansion
    \[
    \omega = \tfrac{1}{2\pi i} {\rm res}_{H_\alpha}(\omega) \wedge d \log\alpha({\bm{z}}) + \eta
    \]
    where $\eta$ is holomorphic along the hyperplane.
\end{itemize}

\noindent For a fixed $\beta_\circ$ the right-hand-side of \cref{keyequation} has poles along the hyperplanes defined by the set of covectors
\[
\mathfrak{A}^{+}_{\beta_\circ} = \{ \alpha\in\mathfrak{A}^{+} \,|\, \langle\alpha,\beta_\circ \rangle^{*} \neq 0\}\,.
\]
The left-hand-side requires more care as different differences could define the same hyperplane (see Figure \cref{differentdifferences}), and this hyperplane may, or may not, match a hyperplane appearing on the right-hand-side.

\begin{figure}
\begin{tikzcd}
\draw[thick] (0,0) - - (-3,1);
\draw[thick,->] (0,0) -- (-1.5,0.5) node[anchor=north east]{\beta_2};
\draw[thick] (0,0) - - (+3,1);
\draw[thick,->] (0,0) -- (+1.5,0.5) node[anchor=north west]{ \beta_1};
\draw[thick] (-3,1) -- (3,1);
\draw[thick,->] (-3,1) -- (-1.5,1) node[anchor=south]{\alpha};
\draw[thick,->] (-3,1) -- (-1.4,1);
\draw[thick,->] (0,1) -- (1.5,1) node[anchor=south]{2\alpha};
\draw[thick,->] (0,1) -- (1.4,1);
\draw[thick,->] (0,0) -- (-0.4,0.6) node[anchor=west]{\beta_\circ};
\draw[thick] (0,0) -- (-0.66,1.0);

\end{tikzcd}
\caption{\label{differentdifferences}}{Example: Different differences $(\beta_\circ-\beta_1\neq \beta_\circ-\beta_2)$ defining the same hyperplane $H_\alpha$}
\end{figure}

To keep track of these different types of structures we decompose, for each $\beta_\circ$, the set of differences $\{\beta - \beta_\circ\}$ into disjoint sets
\begin{eqnarray*}
    \mathfrak{B}_{\beta_\circ}^{\bullet} &=&
    \{ \beta_\circ - \beta \,|\, \exists \alpha\in \mathfrak{A}^{+}_{\beta_\circ} \quad \text{s.t.}\quad \beta_\circ - \beta\sim \alpha \}\,,\\
    \mathfrak{B}_{\beta_\circ}^{\circ} &=&
    \{ \beta_\circ - \beta \,|\, \nexists \alpha\in \mathfrak{A}^{+}_{\beta_\circ} \quad \text{s.t.} \quad\beta_\circ - \beta \sim \alpha \}\,
\end{eqnarray*}
depending on whether the difference defines a hyperplane appearing on the right-hand-side or not.

For \cref{keyequation} to hold, the residues along the hyperplanes appearing on the right-hand-side (which are defined by the set $\mathfrak{A}^{+}_{\beta_\circ}$), must equal the residues along the hyperplanes on the left-hand-side defined by the set $\mathfrak{B}_{\beta_\circ}^{\bullet}$, with the remaining residues along the hyperplanes defined by $\mathfrak{B}_{\beta_\circ}^{\circ}$ being zero. This may be expressed succinctly in terms of a bijection
\[
\rho_{\beta_\circ}\,: \mathfrak{A}^{+}_{\beta_\circ} \longrightarrow 
\displaystyle{{\mathfrak{B}^{\bullet}_{\beta_\circ}}\slash{\sim}}\,.
\]
This maps a covectors $\alpha\in\mathfrak{A}^{+}_{\beta_\circ}$ to the equivalence class $[\beta_\circ-\beta]$ all elements of which define the same hyperplane $H_\alpha\,.$ Since the number of hyperplanes with non-zero residue appearing on each side of \cref{keyequation} must be equal,
\[
\vert \mathfrak{A}^{+}_{\beta_\circ} \vert =
\left\vert 
{{\mathfrak{B}^{\bullet}_{\beta_\circ}}\slash {\sim}}\right\vert\,.
\]
When this is true, \cref{keyequation} holds if and only if the residues along the hyperplanes of the left and right-hand side coincide. That is, first, for all $\beta_\circ\in\mathfrak{B}$\,:
\[
 \sum_{v \in \rho_{\beta_\circ}(\alpha)} k_v \, v = h_\alpha \langle \alpha, \beta_\circ \rangle^\ast \, \alpha \,,\qquad \forall \alpha\in \mathfrak{A}^+_{\beta_\circ}\,.
\]
Here, if \(v=\beta_\circ-\beta\), we write \(k_v:=k_\beta\). Second, for any equivalence class $[v]\in \mathfrak B^\circ_{\beta_\circ}/\!\sim\,$, it must similarly be:
\[
\sum_{v\in [v]} k_v v = 0\,.
\]
We thus arrive at the following:

\medskip

\begin{thmdef}\label{mainresult}
Let $\cech$ be a $\bigvee$-system and $\widetilde{\mathfrak{B}}$ be a collection of covectors in $\widetilde{V}\,$, and $\mathfrak{B}\subseteq V^\ast$ be the projection onto $V^\ast$ of the covectors in $\widetilde{\mathfrak{B}}$ having non-zero component in the extended direction, i.e. in vectors of the form 
$(1,-\beta)\,.$ Assuming the centre-of-mass condition $\sum_\beta k_\beta \beta =0\,,$ the functions $\Frst_\cech$ and $\overset{\star}{\Omega}_{\widetilde{\mathfrak{B}}}$ in \cref{eq:prepotentialcech} and \cref{eq:omegaprepotentialcech2} respectively satisfy the open WDVV equations if and only if the set $\{k_\beta, \beta\in\mathfrak{B}\}$ satisfies the following conditions for all $\beta_\circ\in\mathfrak{B}\,$:
\begin{itemize}
\item[(i)] The map
\begin{equation}\label{conditionA}
\rho_{\beta_\circ} \,:\, \mathfrak{A}^{+}_{\beta_\circ} \longrightarrow \displaystyle{{\mathfrak{B}^{\bullet}_{\beta_\circ}}\slash{\sim}}
\end{equation}
is a bijection;

\item[(ii)] For all $\alpha\in\mathfrak{A}^{+}_{\beta_\circ}\,,$ the residues along the hyperplane $H_\alpha$ satisfy the condition
\begin{equation}\label{conditionB}
 \sum_{v \in \rho_{\beta_\circ}(\alpha)} k_v \, v = h_\alpha \langle \alpha, \beta_\circ \rangle^\ast \, \alpha\,; 
\end{equation}
Here, if \(v=\beta_\circ-\beta\), we write \(k_v:=k_\beta\,\).
\medskip

\item[(iii)] For every equivalence class $[v]\in\mathfrak B^\circ_{\beta_\circ}/\!\sim$, the residue along the corresponding hyperplane vanishes: 
\begin{equation}\label{conditionC} 
    \sum_{v\in [v]} k_v\,v=0\,. 
\end{equation}
\end{itemize}

In such a case, we say that $\widetilde{\mathfrak{B}}$ (or, equivalently, $\mathfrak{B}$) defines an \textit{open $\bigvee$-system} corresponding to the given (closed) $\bigvee$-system $\cech\,$.
\end{thmdef}
\begin{oss}
    These conditions form an over-determined system. Compatibility can impose conditions on the otherwise free data in the original closed $\bigvee$-system.
\end{oss}

\section{Coxeter Group Examples}\label{sec:Coxeterexamples}
As there is no classification of closed $\bigvee$-systems one cannot classify open $\bigvee$-systmes. Here we construct the simplest examples of open $\bigvee$-systems starting from Coxeter configurations, i.e. $\mathfrak{A}=\mathfrak{R}_W$ is the root system for a finite irreducible Coxeter groups $W\,.$ This automatically defines a $\bigvee$-system.

Since the above construction involves the differences $\beta_\circ-\beta$, it is natural to consider $\mathfrak{B}=W\omega_\circ$ to be the Weyl orbit of some weight vector $\omega_\circ\,.$ Thus by Weyl symmetry one does not have to consider all vectors $\beta_\circ$, one can just consider the single case $\beta_\circ=\omega_\circ$ and the other cases follow by Weyl symmetry, and one can also take the constants $k_\beta$ to be all equal. A classical result \cite{Bou,humphreys_1990} then says that $\omega_\circ-w(\omega_\circ)\,, w\in W$ is a sum of roots. However, sums of roots are not (in general roots) and for all these differences to be proportional to a single root places restrictions on both $W$ and the choice of $\omega_\circ\,.$

We recall the definition of a small orbit \cite{serganova}:

\begin{defn}
Let $W$ be the finite, irreducible Weyl group and let $\omega_\circ$ be an arbitrary weight. The $W$-orbit $\{ w(\omega_\circ)\,, w\in W\}$ is small if 
$\omega_\circ-w(\omega_\circ) \in \mathfrak{R}_W$ for all $\omega_\circ\neq \pm w(\omega_\circ)\,.$
\end{defn}

\noindent These were classified (by first showing that $\omega_\circ$ must be a fundamental weight) in \cite{serganova} and the results are summarised in the following table:

\begin{center}
\begin{tabular}{ c|| c l }
 $W$ & $\omega_\circ$ &  \\ 
 \hline
 \hline
 $A_n$ & $\omega_1\,, \omega_n$ & ({\rm by~symmetry})\\ 
 $B_n$ & $\omega_1$ &(and $\omega_3$ for $B_3$)\\
 $D_n$ & $\omega_1$ & (and $\omega_3\,,\omega_4$ for $D_4$)\\
 $G_2$ & $\omega_1$ &\\
\end{tabular}
\end{center}

\noindent We now apply these results to construct open $\bigvee$-systems $\mathfrak{B}=W\omega_\circ$ from the open $\bigvee$-system $\mathfrak{A}=\mathfrak{R}_W$ with $W$ and $\omega_\circ$ from the above table. Since Serganova's result is based on properties of Weyl groups, the Coxeter groups are automatically crystallographic. The ideas can also be extended to the non-crystallographic groups $I_2(N)$ and $H_3\,.$

\subsection{Crystallographic examples}

We use the standard notation for roots and weights \cite{Bou,humphreys_1990}.

\paragraph{ The Coxeter group $A_n$.}
Since every Coxeter group is a (closed) $\bigvee$-system, we automatically have
\[
\Frst_{A_n}(\bm{z})=\tfrac{1}{4} \sum_{\alpha\in \mathfrak{R}_{A_n}} \alpha({\bm{z}})^2 \log \alpha({\bm{z}})
\]
(so $h_\alpha=1$ for all $\alpha\in \mathfrak{R}_{A_n})$. To construct the associated open $\bigvee$-system we require basic properties of roots and weights.

For $A_n$ we take $V$ to be the hyperplane $\sum_{i=1}^{n+1} z_i = 0$ in $\RR^{n+1}$ and with this positive roots
\[
{\mathfrak{A}^{+}}=\mathfrak{R}_{A_n}^{+} = \{ e_i-e_j\,, 1\leq i<j \leq n+1\}
\]
and
\[
\omega_1 = e_1 - \frac{1}{n+1} \sum_{j=1}^{n+1} e_j\,.
\]
Since $W_{A_n}$ acts by permutations, 
\[
W_{A_n}(\omega_1) = \biggl\{ e_i - \frac{1}{n+1} \sum_{j=1}^{n+1} e_j\,,\quad j=1\,,\ldots n+1\biggr\}\,,
\]
and hence
\[
\mathfrak{B}^{\bullet}_{\omega_1} \cong \mathfrak{A}^{+}_{\omega_1} = \{ e_1 - e_i\,, i=2\,,\ldots\,,n+1 \}\,.
\]
Thus, the hyperplanes on each side of \cref{keyequation} match exactly and no equivalence relation is required (and the set $\mathfrak{B}^{\circ}_{\omega_1}$ is empty). Thus in the Theorem/Defintion \ref{mainresult}, equations \cref{conditionA} and \cref{conditionC} are satisfied. Condition \cref{conditionB} (assuming $k_\beta=1$) reduces to the condition $k_\beta=h_\alpha\,.$ With these results
\begin{equation}\label{eq:Ansmallorbit}
    \overset{\star}{\Omega}_{A_n}(x,\bm{z}) = \left.\sum_{i=1}^{n+1} (x - z_i) 
\log(x-z_i)\right|_{\sum z_i = 0}\,.
\end{equation}
These correspond to the solutions found in \cite{openWDVVduality}.

\begin{ex}\label{A3example} $(W=A_3,\omega_\circ=\omega_2)$
This is an example for which $\omega_\circ$ does not generate a small orbit. In the standard basis, $\omega_\circ=\omega_2=\frac{1}{2} (1,1,-1,-1)\,$, and its Weyl orbit consists of the six vectors given by permuting the signs. The non-zero differences are then
\[
\omega_2 - W_{A_3}(\omega_2) = \{ (1,0,-1,0)\,,(1,0,0,-1)\,,(0,1,-1,0)\,,(0,1,0,-1)\} \cup \bigl\{ (1,1,-1,-1)\bigr\}\,.
\]
The elements in the first set are positive roots, while the single element in the second set is not. However,
\[
\mathfrak{A}^{+}_{\omega_2} = \{ (1,0,-1,0)\,,(1,0,0,-1)\,,(0,1,-1,0)\,,(0,1,0,-1)\}\,.
\]
The set of differences can thus be decomposed
\begin{eqnarray*}
    \mathfrak{B}^{\bullet}_{\omega_2} &=&  \mathfrak{A}^{+}_{\omega_2}\,,\\
     \mathfrak{B}^{\circ}_{\omega_1} &=& \bigl\{ (1,1,-1,-1)\bigr\}\,.
\end{eqnarray*}
This, by itself, does not satisfy the conditions in Theorem/Definition \ref{mainresult}. To satisfy them one has to append the zero vector to the set $\mathfrak{B}$ (for details, see the $D_n$-example below). This finally gives the solution

\begin{equation}\label{eq:A3omega2}
    \overset{\star}{\Omega}_{A_3,\omega_2}(x,\bm{z}) = \sum_{\beta\in W_{A_3}(\omega_2)} (x - \beta({\bm{z}}) ) \log (x - \beta({\bm{z}}) ) - 2 x \log x\,.
\end{equation}
\end{ex}

\noindent This example shows that, given a closed $\bigvee$-system $\mathfrak{A}\,$, the open $\bigvee$-system associated to it is not unique (if it exists). 
\paragraph{ The Coxeter group $B_n$.}

The $B_n$ solution
\[
\Frst_{B_n}(\bm{z})=\tfrac{1}{4} \sum_{\alpha\in \mathfrak{R}_{B_n}} h_\alpha \,\alpha({\bm{z}})^2 \log \alpha({\bm{z}})
\]
depends on two parameters:
\[
h_\alpha = \begin{cases}
h_s & {\rm if~}\alpha~{\rm is~a~short~root\,,}\\
h_l & {\rm if~}\alpha~{\rm is~a~long~root\,.}
\end{cases}
\]
In this case, $V=\mathbb{R}^n\,$, the positive roots are:
\[
\mathfrak{A}^{+} = \mathfrak{R}_{B_n}^{+} = \{e_i\,,i=1\,,\ldots\,,n\} \cup \{ e_i\pm e_j\,, 1 \leq i<j\leq n\}\,,
\]
and $\omega_1=e_1\,.$ Then
\[
\mathfrak{A}^{+}_{\omega_1} = \{ e_1\} \cup \{e_1 \pm e_j\,, 2 \leq j \leq n\}
\]
(so $|\mathfrak{R}^{+}_{\omega_1}|=2n-1)$. Since $-\omega_1$ lies in the orbit $W\omega_1$ we obtain
\[
\mathfrak{B}^{\bullet}_{\omega_1} = \{ 2 e_1\} \cup \{e_1 \pm e_j\,, 2 \leq j \leq n\}\,.
\]
Thus the hyperplanes on each side of \cref{keyequation} match exactly and no equivalence relation is required (and the set $\mathfrak{B}^{\circ}_{\omega_1}$ is empty). Thus in the \cref{mainresult}, equations \cref{conditionA} and \cref{conditionC} are satisfied. 

The condition \cref{conditionB}, on the other hand, is satisfied provided that:
\[
2k=h_s\,,\qquad k=h_l\,.
\]
This places a restriction on the otherwise free constants, namely $h_s=2h_l\,.$ With $k=1\,,$ we have  $h_s=2\,,h_l=1\,.$ With these results, the open WDVV solution is:
\begin{equation}\label{eq:Bnsmall}
    \overset{\star}{\Omega}_{B_n}(x,\bm{z}) = \sum_{i=1}^{n} (x \pm z_i) 
\log(x\pm z_i)\,.
\end{equation}

\paragraph{ The Coxeter group $D_n$.}

For $D_n$ one has
\[
\mathfrak{A}^{+} = \mathfrak{R}^{+}_{D_n} = \{ e_i\pm e_j\,|\,1\leq i < j \leq n\}
\]
and $\omega_1=e_1\,,$ so $W\omega_1=\{\pm e_i\,, i = 1\,,\ldots\,,n\}\,.$ As in the $B_n$ example, $-\omega_1 \in W\omega_1$ but unlike $B_n$, the difference $\omega_1-(-\omega_1)=2\omega_1$ is not proportional to a root, so the set $\mathfrak{B}^\circ_{\omega_1}$ is non-empty. However, \cref{keyequation} has a non-zero residue along the hyperplane $H_{\omega_1}$ and hence cannot hold.

To overcome this problem we append the zero-vector $\bf{0}$ to the set $\mathfrak{B}$, so
\[
\mathfrak{B} = \{\pm e_i\,, i = 1\,,\ldots\,,n\}\cup\{ {\bf{0}}\}
\]
with corresponding constants $k$ and $k_{\bf{0}}$ respectively.
The conditions in Theorem/Definition \ref{mainresult} have to hold for all $\beta_\circ\in \mathfrak{B}$ but they are vacuous for $\beta_\circ={\bf{0}}$ and so, without loss of generality, we take $\beta_\circ=e_1\,.$

Then
\[
\mathfrak{A}^{+}_{\omega_1} = \{ e_1 \pm e_j\,, 2\leq j \leq n\}
\]
(so $|\mathfrak{A}^{+}_{\omega_1}|=2(n-1)$) and the set of non-zero differences
\[
\omega_1 - W_{D_n}(\omega_1) = \{ e_1\pm e_j, \,, 2\leq j \leq n\} \cup\{2 e_1\}\cup\{e_1\}
\]
decomposes into the disjoint sets
\begin{eqnarray*}
    \mathfrak{B}^{\bullet}_{\omega_1} &=& \{ e_1\pm e_j, \,, 2\leq j \leq n\}\,,\\
    & = & \mathfrak{A}^{+}_{\omega_1}\,,\\
     \mathfrak{B}^{\circ}_{\omega_1} &=& \{2 e_1\}\cup\{e_1\}\,.
\end{eqnarray*}
Condition \cref{conditionA} is thus satisfied, and \cref{conditionB} gives $k=1$ (on setting $h_\alpha=1$ for all $\alpha$). Condition \cref{conditionC} now yields the condition $2k+k_{\bf{0}}=0\,.$ With these results
\begin{equation}\label{eq:Dnsmall}
    \overset{\star}{\Omega}_{D_n}(x,\bm{z}) = \sum_{i=1}^{n} (x \pm z_i) 
\log(x\pm z_i) - 2 x \log x\,.
\end{equation}

Notice that $\overset{\star}{\Omega}_{D_3}$ and the solution $\overset{\star}{\Omega}_{A_3,\omega_2}$ in \cref{A3example} are the same up to a linear change of the $z$-variables. This correspond to the standard Weyl group isomorphism $W_{D_3}\cong W_{A_3}\,$.

\paragraph{ The Coxeter group $G_2$.}

This example will be included in the next section as a special case of a dihedral group. Since $G_2$ has long and short roots one has different constants $h_l$ and $h_s$ (as in the $B_n$-case). The construction places restriction on these otherwise free data.

\subsection{Non-crystallographic examples}

The crystallographic examples above rested on the classification of small orbits where differences between two non-proportional vectors in the Weyl orbit was a root.  This condition can be extended to cases where the differences between two non-proportional vectors in the Weyl orbit is {\sl proportional} a root, and this generalization includes non-crystallographic examples. It results from the following simple geometric fact: if the difference is proportional to a root (see Figure \ref{noncrystal}) then the constant of proportionality is 
$\pm \langle\alpha_{12},\omega_1\rangle^{*}\,.$ 

\begin{figure}
\begin{tikzcd}
\node[] at (-3,0) {};
\node[] at (4.5,0.8) { \omega_2-\omega_1 = \langle\alpha_{12},\omega_1\rangle^{*}\alpha_{12}};
    \draw[thick] (0,0) - - (2,3);
    \draw[thick, ->] (0,0) - - (1,1.5) node[anchor= east] {\omega_2};
    \draw[thick] (0,0) - - (3,2);
    \draw[thick,->] (0,0) - - (1.5,1)  node[anchor= west] {\omega_1};
    \draw[thick] (3,2) - - (2,3) ;
    \draw[thick,->] (3,2) -- (2.4,2.6);
    \draw[thick, ->] (3,2) - - (2.5,2.5) node[anchor=north east] {\lambda\alpha_{12}\!\!};
\end{tikzcd}
\caption{\label{noncrystal}}{Geometry of non-crystallographic root systems}
\end{figure}

\paragraph{ The Coxeter group $I_2(N)$.}

Consider a regular $N$-gon with a vertex at $(1,0)\,.$ We can take as weight vectors (using $\mathbb{C}$, instead of $\mathbb{R}^2$, and the Euclidean metric $dz\, d\bar{z}$) the position vectors of the vertices
\[
w_p = e^{2 \pi p i/N}\,, \qquad p=0\,,\ldots\,,N-1\,,
\]
so $w_\circ=1$ is the original vertex. As roots, one can take
\[
\alpha_p = \sqrt{2} \,e^{i(\pi/2 + \pi p/N)}\,, \qquad p=0\,, \ldots 2N-1
\]
with positive roots corresponding to $p=0\,,\ldots\,,N-1\,.$ The roots have been normalised to have length-squared equal to two. A simple calculation gives
\[
(w_\circ-w_p)({\bm{z}}) =  \langle\alpha_p,w_\circ\rangle^{*} \, \alpha_p({\bm{z}})\,,\qquad p=1\,,\ldots\,,N-1\,.
\]
Conditions \cref{conditionA}, \cref{conditionB} and \cref{conditionC} are then clearly satisfied, giving the solution
\begin{equation}\label{eq:dihedral}
    \overset{\star}{\Omega}_{I_2(N)} (x,\bm{z})= \sum_{p=0}^{N-1} ( x - w_p({\bm{z}})) \log( x - w_p({\bm{z}}))\,.
\end{equation}
Since $W_{G_2}\cong W_{I_2(6)}\,$, we get the corresponding $G_2$ open WDVV solution as a special case in the dihedral family.

\paragraph{ The Coxeter group $H_3$.}

The roots of the Coxeter group $H_3$ may be taken to be the cyclic permutations of the vectors (where 
$\tau =\frac{1+\sqrt{5}}{2}$ is the golden ratio):
\[
(\pm\sqrt{2},0,0) \quad {\rm and~}
\left\{
\begin{array}{c}
\frac{1}{\sqrt{2}} ( \pm \tau,\pm 1,\pm \tau^{-1})\,,\\
\frac{1}{\sqrt{2}} ( \pm 1,\pm \tau^{-1},\pm \tau)\,,\\
\frac{1}{\sqrt{2}} ( \pm \tau^{-1},\pm \tau,\pm 1)\\
\end{array}
\right.
\]
yielding 30 roots. The corresponding weight vectors in the smallest orbit are the cyclic permutations of the vector
\[
(\pm 1, \pm\tau,0)\,,\quad (0,\pm 1,\pm\tau), \quad (\pm\tau,0,\pm 1)\,.
\]
yielding 12 weights. A direct calculation shows that the difference between two non-proportional weights is always proportional to a root, with the constant of proportionality being $\pm\sqrt{2}$ or $\pm \sqrt{2}\,\tau\,.$ A similar calculation to Example \ref{A3example} or the $D_n$ family yields
\[
\overset{\star}{\Omega}_{H_3} = \sum_{i=1}^{12} (x - w_i({\bm{z}})) 
\log(x -w_i({\bm{z}}))  - 2 x \log x\,.
\]

\begin{oss}
In the $B_n$-case, the $\bigvee$-conditions place no restriction on the lengths of the two constants $h_s$ and $h_l\,.$ However, the open $\bigvee$ conditions are only satisfied if these scales are related. By simple scaling one obtains a $\bigvee$-system where all the roots are the same length. This requirement also come from almost-duality: the almost-dual $B_n$ superpotential constructed by \cite{Dub04} has, by construction, all the roots the same length. The same condition occurs if one starts with the $G_2$ root system - applying the open $\bigvee$-conditions forces the covectors in configuration to all have the same lengths - as in the dihedral examples above.
\end{oss}

\section{Superpotentials from Open WDVV equations}\label{sec:superpotentials}
Recall, from Theorem \ref{HuwitzTheorem}, that for an underlying Frobenius manifold, that the relation
\[
\overset{\star}{\Omega}_x = \log\lambda
\]
holds. Given the form of $\overset{\star}{\Omega}_{\widetilde{\mathfrak{B}}}\,$, it immediately follows that
\[
\lambda_{\mathfrak{B}}(x) = \prod_{\beta\in\mathfrak{B}}
 (x - {{\beta}}({\bm{z}}))^{k_\beta}\,.
 \]
Here we use this result to recover the standard superpotentials for the Coxeter groups examples constructed in the last section. These are, of course, already known, but the construction can also be used to superpotentials that generate dual-type Frobenius structures but not Frobenius manifolds. In particular, in all the examples in \cref{sec:Coxeterexamples}, $k_\beta=1$ except for, at most, $k_{\bm{0}}$ when the addition of the zero vector to the set $\mathfrak{B}$ has been necessary. 

When the $\bigvee$-system is a Coxeter system associated to the Coxeter group $W\,$, the superpotential is $W$-invariant. Furthermore, by Chevalley's Theorem, the ring of polynomial $W$-invariants admits a minimal system of generators of positive degree $\sigma_1(\bm{z}),\dots, \sigma_{\rank W}(\bm{z})\in \CC[\bm{z}]^W\,$. It follows that we can write, when $\bm{0}\notin\mathfrak{B}\,$:
\[
\prod_{\beta\in\mathfrak{B}}  (x - \beta({\bm{z}})) = 
x^{|\mathfrak{B}|} + \sum_{i=2}^{|\mathfrak{B}|} s_i({\bm{\sigma}({\bm{z}})}) x^{|\mathfrak{B}|-i}\,,
\]
where the $\{s_i\}_{i=2,\dots, \abs{\mathfrak{B}}}$ are homogeneous $W$-invariant polynomials, and can therefore be written as polynomial functions in the basic invariants. Note, in particular, that $s_1=0\,$ due to the centre-of-mass condition, and that $s_2$ needs to be proportional to the unique degree-two invariant polynomial. With this, one can recover the standard superpotentials for the $A_n\,,B_n$ and $D_n$ Saito Frobenius manifolds \cite{Duborbitspaces}.

\begin{ex} {(The $A_n\,,B_n$ and $D_n$ Coxeter groups)}
\begin{eqnarray*}
    \lambda_{A_n}({x}) & = & {x}^{n+1} + \sum_{i=2}^{n+1} s_i({\bm{z}})\, x^{n+1-i}\,,\\
    &&\\
    \lambda_{B_n}({x}) & = & {x}^{2n} + \sum_{i=1}^{n} s_{2i}({\bm{z}})\, x^{2(n-i)}\,,\\
    &&\\
    \lambda_{D_n}({x}) & = & \frac{1}{x^2} \left({x}^{2n} + \sum_{i=1}^{n} s_{2i}({\bm{z}})\, x^{2(n-i)}\right)\,.
\end{eqnarray*}
The $D_n$-superpotential differs from the one obtained by the unfolding of the $f_{D_n}(x,y) = x^{n-1} + x y^2$ singularity, but gives rise to the same Frobenius manifold \cites{Bertolaphd,Zuo07}. Note also that the pole terms has its original in the fact that the set $\mathfrak{B}^\circ_{\omega_1}$ is non-empty. The non-small-orbit solution for the Weyl group $A_3$ in \cref{A3example} gives the same superpotential as $\lambda_{D_3}\,$, as we have discussed. 
\end{ex}

\noindent The superpotentials for the non-crystallographic examples require more care.

\begin{ex}{(The dihedral group $I_2(N)$)}
Since the invariant polynomials for the dihedral group $I_2(N)$ are generated by the polynomials $s_2(\bm{z}):=z_1^2+z_2^2$ and $s_N(\bm{z}):=z_1^N+z_2^N$ it follows that the superpotential is
\[
\lambda_{I_2(N)}(x) = x^N + \sum_{2a+b=N} c_{ab}\, s_2({\bm{z}})^a \, x^b + s_N({\bm{z}})\,,
\]
for various explicit constants $c_{ab}\,.$ This clearly sits inside the superpotential $\lambda_{A_{N-1}}$ and this corresponding to the classical embedding $I_2(N) \hookrightarrow A_{N-1}$ obtained via folding arguments \cite{ZuberotherCoxeter}. For example, for $I_2(8)$ we obtain the superpotential
\[
\lambda_{I_2(8)}(x) = x^8 + t_2 \,x^6 + \frac{5}{16} t_2^2\, x^4 + \frac{1}{32} t_2^3 \,x^2 + \left( t_1 + \frac{1}{2048} t_2^4\right)\,.
\]
where the ambiguity in the definition of the highest degree invariant polynomial has been used to ensure that the $\{t_i\}$ are the flat coordinates for the intersection form. This differs from the more usual superpotential
\[
\lambda(x) = x^8 + t_2 \, x^6 +\left( t_1 + \frac{27}{512} t_2^4\right)\,,
\]
which gives the same Frobenius manifold. Note that the factorization of this later superpotential in terms of weights in the set $\mathfrak{B}$ is lost.
\end{ex}

\begin{ex}{(The $H_3$-Coxeter Group)} The same idea may be applied to the $H_3$-superpotential, given the degrees of the polynomial invariants are $2\,,\,6$ and $10\,$ (this has also been studied independently in \cite{FM}). This gives a superpotential contained in the $D_6$-family and corresponds to the classical embedding $H_3 \hookrightarrow D_6\,$ \cite{ZuberotherCoxeter}. 

A curious embedding of $H_3$ as a submanifold of $A_9$ is also possible via the superpotential
\[
\lambda(x) = \frac{1}{10} x^{10} +t_3 \, x^8 + (t_2 - 16 t_3^3) \, x^4 + \left( t_1 - 16 t_2 t_3^2 + \frac{384}{5}
t_3^5\right)\,.\]
\noindent This is curious in the sense that it does not come from a classical folding argument. This result was found by direct calculation.
\end{ex}

\section{Conclusion}
The examples constructed in the last section were all based on Coxeter groups and hence all have 
$\mathfrak{A} = \mathfrak{R}_W\,.$ In addition, as well as being derived from an orbit-space/Saito-construction, they all can also be derived from Hurwitz-Frobenius manifolds and hence fall under the construction given in Theorem \ref{HuwitzTheorem}. Other classes of examples may also be found by introducing multiplicities and/or generalized root systems.

\begin{ex}[The series $A(m,n)$ \cites{serganova,FV08,SV14}]
Consider the data:
\begin{eqnarray*}
\mathfrak{A} & = & \{ \alpha_{ij} = e_i - e_j\,, i\neq j\,,i,j=0, \ldots\,, n+m\}\,,\\
\varepsilon_i & = & \begin{cases}
    +1 & i=0\,,\ldots\,, m \,,\\
    -1 & i=m+1\,,\ldots\,,n+m\,,\\
    \end{cases}\\
    g & = & \left.\sum_{i=0}^{n+m} \varepsilon_i \,dz_i^2 \,\right|_{\sum \varepsilon_j z_j = 0}\,,
\end{eqnarray*}

\noindent which has a small orbit \cite{SS21}
\[
\mathfrak{B} = \biggl\{ w_i = e_i + \frac{1}{n-m+1} \sum_{r=0}^{n+m} \varepsilon_r \, e_r\,,\quad i=0,\dots, m+n\,\biggr\}\,,
\]
and hence $w_i({\bm{z}}) = z^i\,.$
This defines a Hurwitz space with superpotential and dual-type open WDVV solution:
\[
\lambda_{A(m,n)} (x)= \left.\prod_{i=0}^{m+n} (x-z_i)^{\varepsilon_i}\right|_{\sum \varepsilon_j z_j = 0}\,,
\qquad
\overset{\star}{\Omega}_{A(m,n)}(x,\bm{z}) = \left.\sum_{i=0}^{m+n} \varepsilon_i (x-z_i) \log(x-z_i)\right|_{\sum \varepsilon_j z_j = 0}\,.
\]
The underlying dual-type WDVV solution is \cite{rileyeaw}:
\[
\Frst_{A(m,n)}(\bm{z})=\tfrac14\sum_{i\neq j}\varepsilon_i\varepsilon_j\,(z_i-z_j)^2\log(z_i-z_j)\,.
\]
A similar result holds for the series $B(m,n)\,.$ 
\end{ex}

More generally, one may consider superpotentials with zeros and poles of higher-order,
\[
\lambda(x) = \left.\prod_{i=0}^{m+n} (x-z_i)^{k_i}\right|_{\sum k_j z_j = 0}\,,\qquad k_i\in\mathbb{Z}^\ast\,.
\]
For simple zeros this still defines a Hurwitz-Frobenius manifold (including for higher-order poles), but if one introduces multiplicities for the zeros, the metric given by the residue formula \cref{eq:residueformulaeLG} is not flat. Geometrically these multiplicities define a discriminant submanifold in the space and induced Saito metric is not flat. However, the induced intersection form on the submanifolds is flat \cite{Dubrovin1996}, and hence one obtain the almost-dual structures defined on the discriminant \cite{StrachanSubMan}. The associated dual-type WDVV solution was also explicitly computed in \cite{rileyeaw}. The associated $\Omega$ then gives a solution to the open WDVV equations. While the geometry depends on the $k_i\in\mathbb{Z}^\ast\,$, the algebraic properties do not require this, and hence one can obtain solutions in the same class that are not directly related to Frobenius manifolds.

More algebraically, subsystems of $\bigvee$-systems and restrictions to subspaces preserve the $\bigvee$-conditions and one area of future research would be to extend these results to open $\bigvee$-systems. Trigonometric and elliptic $\bigvee$-systems have also been studied, and hence the results in this paper could also be extended to open trigonometric and elliptic $\bigvee$-systems. Examples in these families, again based on Hurwitz spaces, have already been constructed. A better understanding of the relationship between certain rational and trigonometric $\bigvee$-systems is also required \cite{FKS24}.

One of the reasons why open WDVV theory was developed was the lack of a non-degenerate metric on the space $\widetilde{V}$ and hence one has vector-valued potential functions (as there is no metric to lower and index). Interestingly, the $\mathfrak{B}$-system does provide a non-degenerate metric (the analogue of the intersection form) on $\widetilde{V}$. Consider
\[
    (\widetilde{X},\widetilde{X})  = \sum_{\tilde{\beta}\in\widetilde{\mathfrak{B}}} 
k_{\widetilde{\beta}} \widetilde{\beta}(\widetilde{X})^2
\]
or, in components,
\begin{eqnarray*}
    \left( (x,{\bm{z}}), (x,{\bm{z}}) \right) & = & 
    \sum_{\beta\in\mathfrak{B}} k_\beta (x-\beta({\bm{z}}))^2\,,\\
    & = & h \,x^2 + \sum_{\beta\in\mathfrak{B}} k_\beta \beta({\bm{z}})^2\,
\end{eqnarray*}
on using the centre-of-mass condition. In the Coxeter examples the last term is proportional to the metric on $V$, from the uniqueness of degree-two invariant polynomials. This suggests one could develop a Saito-type construction of open WDVV-equations where the Saito metric is degenerate on $\widetilde{V}$ but non-degenerate on $V\,.$

Finally, it would be of interest to obtain examples of open $\bigvee$-systems corresponding to the various known $\bigvee$-systems that do not arise from Coxeter systems via taking restriction or subsystems, and to extend the theory to higher-rank extensions. For instance, the other families of solutions to the open WDVV equations in type-$A$ computed in \cite{openWDVVduality} suggest that a similar root-theoretic construction might be possible in the trigonometric and elliptic cases.

\begin{acknowledgements}
   A. P. was supported by a Ph.D. studentship of the EPSRC Doctoral Training Partnership (EP/W524359/1). We would like to thank Misha Feigin and the anonymous referee for their comments on the original draft.
\end{acknowledgements}

\noindent{\bf Conflict of Interest:} On behalf of all authors, the corresponding author states that there is no conflict of interest.

\noindent{\bf Data Availability:} Not applicable as article has no associated data.

\phantomsection
\addcontentsline{toc}{section}{References}
\bibliography{biblio} 

\end{document}